\documentclass[draft, onecolumn, 11pt]{IEEEtran}

\usepackage[utf8]{inputenc} 
\usepackage[T1]{fontenc}
\usepackage{lmodern}
\usepackage{cite}
\usepackage[cmex10]{amsmath} 
\usepackage{amssymb,amsthm}
\usepackage{combelow}
\usepackage[final, hidelinks]{hyperref}
\usepackage{authblk}

\newtheorem{theorem}{Theorem}[]
\newtheorem{definition}[theorem]{Definition}

\newtheorem{lemma}[theorem]{Lemma}
\newtheorem{proposition}[theorem]{Proposition}
\newtheorem{corollary}[theorem]{Corollary}
\newtheorem{remark}[theorem]{Remark}

\newcommand{\F}{\mathbb{F}}

\begin{document}

\title{Algebraic Geometry Codes Approach the Half-Singleton Bound with Constant Field Size} 

\author{Neehar Verma,\thanks{This work was supported by the MSCA Doctoral Networks 2021, HORIZON-MSCA-2021-DN-01 (ENCODE, grant \#101072316).} 
\thanks{This research project was initiated while C. Hollanti was visiting at the Simons Institute for the Theory of Computing at the University of California, Berkeley. }\thanks{R. Tajeddine would like to thank the Aalto Science Institute and the Department of Mathematics and Systems Analysis at Aalto University for the funding and hospitality toward her research visits.}
Camilla Hollanti,  and 
Razane Tajeddine
}

\affil{Aalto University, Finland}
\maketitle

\begin{abstract}
We study linear codes for insertion and deletion (insdel) errors through the lens of evaluation codes. We develop a general framework for analyzing random puncturings of evaluation codes, where the edit distance is controlled by only the size of the evaluation domain and the maximum number of zeros of a nonzero function in the underlying function space. Our proof generalizes the results of \cite{con2024randomRS}, and simultaneously simplifies their arguments by avoiding an in-depth analysis of longest common subsequences. We demonstrate the applicability of our core theorem by instantiating it with random puncturings of Reed--Muller codes. We then recover the result that random Reed--Solomon codes approach the half-Singleton bound over linear-sized fields while also improving the dependence on the additive gap $\varepsilon$ from $2^{O(1/\varepsilon^2)}$ to $2^{O(1/\varepsilon)}$. Finally, by applying the framework to algebraic geometry codes arising from asymptotically good towers of function fields, we show that there exist randomized families of structured linear codes over constant-sized fields that approach the half-Singleton bound. 
\end{abstract}

\section{Introduction}
Error-correcting codes are imperative in ensuring reliable communication and data storage. Classically, errors are typically modeled as erasures or substitutions which are measured using the Hamming distance. This setting has received vast attention ever since the seminal works by Shannon \cite{shannon1948mathematical} and Hamming \cite{hamming1950error} and is by now very well understood. Today, many algebraic and probabilistic constructions of codes are known for this classical model, achieving strong rate-distance tradeoffs together with efficient encoding and decoding algorithms. 

Insertion and deletion errors, commonly referred to as \emph{insdel} errors and measured by the edit distance, are substantially more tricky to deal with as they destroy the alignment between transmitted and received coordinates and may even alter the length of the transmitted message. The study of codes that can correct against such errors was initiated in the work of Levenshtein \cite{levenshtein1966binary} who studied the optimality of the binary codes of Varshamov and Tenengolts \cite{varshamov1965codes} in the context of insdel errors. After a substantive chronological gap, the first instance of positive-rate binary codes capable of correcting an asymptotically positive fraction of insdel errors was constructed in \cite{schulman1999asymptotically}. 

Recently there has been renewed interest in this error model due to the arrival of reliable DNA synthesis and sequencing technologies --- which are inherently susceptible to insertions and deletions --- and the consequent development of highly space-efficient DNA-based archival data storage techniques \cite{goldman2013towards}. 
For a recent review on DNA-based data storage systems and associated coding theoretic problems the reader may refer to \cite{milenkovic2024review}, and for a complementary information-theoretic perspective they may refer to \cite{viterbo2023review}.

In the case of binary codes, there has been staggering progress \cite{levenshtein2002bounds, guruswami2021explicit, sima2020optimal, sima2020systematic, cheng2018deterministic, haeupler2019optimal, guruswami2016efficiently, guruswami2022zero, yasunaga2024improved, qiu2026binary} in better understanding the capabilities, limitations, and rate-distance tradeoff for the insdel model. For larger, yet constant-sized alphabets the most insightful result is that of Haeupler and Shahrasbi \cite{haeupler2021synchronization-org} who defined synchronization strings and showed that they yield explicit and efficiently decodable codes over constant-sized alphabets that correct a $\delta$-fraction of insdel errors while achieving rate $1-\delta-\varepsilon$, essentially attaining optimality with respect to the \emph{Singleton bound}. Here, and for the remainder of this paper, $\varepsilon$ represents the optimality gap, i.e., the additive gap from achieving the half-Singleton bound with equality.   

The reader may notice that up until now all the cited works on insdel codes are concerned with non-linear codes. It is not a coincidence that these works have drifted in this direction, but a consequence of the restriction that linearity imposes on the maximum rate achievable by a code designed for insdel errors. In particular, it was shown in \cite{abdel2007linear} that a linear code capable of correcting a single deletion can have a rate at most $1/2$, far from the optimal rate of $1$ achievable by non-linear codes. Despite this shortcoming, linearity is yet desirable as it gives us the benefit of compact representation and encoding using a generator matrix, and it underlies many important algebraic code families, including Reed--Solomon and algebraic geometry codes. We will use the standard notation $n$ and $k$ to represent the length and dimension of a linear code, while $q$ will be a prime power representing the size of the underlying finite field.

So, to address the question of the existence of good linear codes for insdel errors, Cheng, Guruswami, Haeupler and Li \cite{cheng2020efficient} established two upper bounds on the rate of a linear code. The alphabet-dependent \emph{half-Plotkin bound} which says that any linear code over $\F_q$ that can correct a $\delta$ fraction of insdel errors has a rate at most $\frac{1}{2}\left(1 - \frac{q}{q-1}\delta\right) + o(1)$, and the alphabet-independent \emph{half-Singleton bound} which says that any linear code that can correct a $\delta$ fraction of insdel errors has a rate at most $\frac{1-\delta}{2}+o(1)$. In other words, an $[n,k]$ linear code can correct at most $n-2k+1$ insdel errors. 
These results apply universally, and yet beg the question of whether structured families of linear codes are capable of achieving the half-Singleton bound. In particular, can classical algebraic code families achieve comparable guarantees against insdel errors? Obtaining such structured linear codes over constant-sized fields has been identified as an important goal in the literature \cite{cheng2023linear}.

\subsection{Reed--Solomon codes}
A natural structured family to study in this setting is Reed--Solomon (RS) codes. For an ordered set of distinct evaluation points $\alpha_1,\dots,\alpha_n \in \F_q$, an $[n,k]_q$ Reed--Solomon code is defined as
$$\mathrm{RS}_{n,k}(\alpha_1,\dots,\alpha_n)
= \lbrace (f(\alpha_1),\dots,f(\alpha_n)) \mid f \in \F_q[X],\ \deg(f)< k \rbrace.$$
RS codes are maximum distance separable (MDS) in the Hamming metric and admit efficient algebraic encoding and decoding algorithms. Their strong algebraic structure makes it natural to ask whether they can also achieve good guarantees against insdel errors. The study of the performance of RS codes for deletions was initiated in \cite{safavi2002traitor} and \cite{wang2004deletion}. Several works have provided constructions of two-dimensional RS codes that are optimal with respect to the half-Singleton bound, i.e., they are capable of correcting $n-3$ insdel errors \cite{duc2019explicit, liu20212}. In \cite{con2023reed} it is shown that  for an $[n,2]_q$ RS code to be optimal the field size must be $q = \Omega(n^3)$ and \cite{con2023optimal} provides an explicit construction with $q = O(n^3)$, thereby completely characterizing RS codes that achieve the half-Singleton bound. The explicit construction is made even more appealing by the linear time  deletion-only decoder of \cite{singhvi2026decoder} and by the average linear time permutation insdel decoder of \cite{zhang2026perminsdel}. On the other hand, a complete characterization of two-dimensional RS codes that cannot correct even one insdel error is given in \cite{beelen2025fulllengthratehalf}.

The two-dimensional setting is therefore very well understood, but in the higher-dimensional regime much less is known. Con, Shpilka, and Tamo showed that RS codes can attain the half-Singleton bound exactly \cite{con2022explicit}, although the size of the required field is exponentially large in the constant-rate regime. To address the large field size, Con, Guo, Li, and Zhang showed that random RS codes approach the half-Singleton bound over linear-sized fields \cite{con2024randomRS}. In particular, for a sufficiently large absolute constant $c$, they show that
$$q\geq n+2^{c/\varepsilon^2}k$$
is sufficient for a random $[n,k]_q$ RS code to correct
$$(1-\varepsilon)n-2k+1$$
insdel errors with high probability.

Their proof is inspired by recent progress on list decoding of random RS codes \cite{brakensiek2023generic, GZ23, alrabiah2023randomly}, where the main technical problem is reduced to showing that particular structured matrices depending on randomly chosen evaluation points --- which capture the essence of list-decoding capacity --- are non-singular. Con, Guo, Li, and Zhang \cite{con2024randomRS} followed a similar approach for the insdel setting. They derive a sufficient algebraic condition in which the relevant matrices depend on pairs of possible common subsequences, and their proof requires a careful structural analysis of the longest common subsequence (LCS) together with a specially designed exposure procedure for the evaluation points. An efficient and general insdel decoder for $k$-dimensional RS codes is given in \cite{banerjee2025decoding} by leveraging list recovery. However, this decoder does not operate efficiently close to the half-Singleton bound and is only capable of correcting  $t$ insdel errors when $tk = O(n)$.
So, despite the strong distance guarantees, efficient decoding up to the half-Singleton bound remains open for general dimensions.

\subsection{Our results}
In this work, we show that the mechanism underlying the result of Con et al.\ in \cite{con2024randomRS} can be made substantially simpler and, more importantly, applies to a much broader class of evaluation codes.

Let $\mathcal{X}$ be a finite evaluation domain and let $\mathcal{F}$ be a $k$-dimensional vector space of functions from $\mathcal{X} \to \F_q$. Suppose that every nonzero function in $\mathcal{F}$ vanishes on at most $Z$ points of $\mathcal{X}$. Choosing distinct points $P_1, \dots, P_n \in \mathcal{X}$ uniformly at random gives the evaluation code
$$C_{\mathcal{F}}(P_1,\dots,P_n) =
\lbrace (f(P_1),\dots,f(P_n)) \mid f \in \mathcal{F} \rbrace.$$
We give a general edit distance guarantee for such codes using only the zero parameter $Z$ and the size of the evaluation domain.

The proof is based on a simple row-exposure argument. A long common subsequence between two distinct codewords gives rise to a matrix that must fail to have full column rank. When its rows are exposed sequentially, each new row contains a fresh evaluation point. Conditioned on all previously exposed points, failure of this row to increase the rank forces the fresh evaluation point to be a zero of a nonzero function in $\mathcal{F}$. Hence, each such failure occurs with probability at most
$$\frac{Z}{|\mathcal{X}|-n+1}.$$
A rank-deficient matrix requires many such failures, making the bad event exponentially unlikely. A union bound over the possible pairs of subsequences then gives the desired insdel error correction guarantee. Thus, after the standard reduction from edit distance to LCS, the proof requires only a fresh evaluation point and a bound on the number of zeros of a nonzero function, thereby avoiding the more involved LCS analysis used in the previous RS proof.

First, we obtain new consequences for Reed--Muller (RM) codes. In particular, for first-order binary RM codes, when $m=o(n)$ and $n=o(2^m)$, we show that a random ordered length-$n$ puncturing of $\mathrm{RM}_2(m,1)$ corrects a constant fraction of insdel errors; more precisely, it corrects any fixed $\delta<0.0376$ fraction of errors with probability $1-2^{-\Omega(n)}$. This complements the recent work of Qiu et al. \cite{qiu2026binary}, who obtain insdel error correction guarantees for specially chosen linear subcodes of binary first-order RM and simplex codes and subsequently puncture these subcodes to improve their rate. In contrast, our result requires no restriction to a subcode and gives a high-probability guarantee for random puncturings of the full first-order RM code.

Next, for RS codes we simply have $Z=k-1$. Our general theorem therefore immediately recovers the linear field size result of Con et al., while also improving its dependence on $\varepsilon$. We show that
$$q\geq n+2^{4/\varepsilon}(k-1)$$
is sufficient for a random RS code, with probability at least $1-2^{-n}$, to correct
$$
n-2k+1-\lfloor\varepsilon n\rfloor
$$
insdel errors. Hence, the dependence on the additive gap $\varepsilon$ from the half-Singleton bound improves from $2^{O(1/\varepsilon^2)}$ to $2^{O(1/\varepsilon)}$.

Finally, the main benefit of our generalization is that it applies directly to algebraic geometry (AG) codes. Let $\mathcal X$ be an algebraic curve over $\mathbb F_q$ and let $G$ be a divisor. Every nonzero function in the Riemann--Roch space $\mathcal L(G)$ has at most $\deg(G)$ zeros away from the support of $G$. Therefore, our general theorem applies immediately with $Z =\deg(G)$. No new LCS analysis or AG-specific probabilistic argument is required. 

This extension allows us to overcome a basic limitation of RS codes. Over $\F_q$, RS codes have only $q$ possible evaluation points, forcing the field size to grow with the block length. In contrast, algebraic curves in asymptotically good towers can provide an unbounded number of rational evaluation points over a fixed field. Using such towers, we obtain randomized families of linear AG codes over constant-sized fields that approach the half-Singleton bound. More precisely, for every $\delta,\varepsilon>0$, there exists a finite field $\F_q$, whose size is independent of the block length, and a family of randomly punctured AG codes over $\F_q$ that, with high probability, can correct a $\delta$-fraction of insdel errors while achieving
$$R\geq \frac{1-\delta}{2}-\varepsilon.$$

General random linear codes with these parameters were already known to exist. Our result shows that the same near-optimal tradeoff is achieved by a structured algebraic family over a constant-sized field. Thus, in the randomized algebraic setting, we achieve the parameter goal of constructing structured linear codes over constant-sized fields that approach the half-Singleton bound.

There is also a close parallel with recent progress in list decoding. Results on randomly punctured RS codes were subsequently extended to AG codes in order to obtain constant field size list-decoding results \cite{brakensiek2024ag}. Our work is inspired by this progression, but does not use these list-decoding results. Rather, our general framework shows directly that the same RS to AG transition also occurs for insertion and deletion errors.

To summarize, we give a simple general framework for analyzing random evaluation codes against insdel errors. For RM codes, the framework gives new insdel error correction guarantees for random puncturings, including a constant-fraction guarantee for first-order binary RM codes. This complements the recent work of Qiu et al. \cite{qiu2026binary} on insdel-robust subcodes of binary first-order Reed--Muller codes, whereas our result gives a constant-fraction guarantee for random puncturings of the full first-order Reed--Muller code. For RS codes, the framework gives a substantially shorter proof of the result of Con et al.\cite{con2024randomRS} and improves the dependence on $\varepsilon$ from $2^{O(1/\varepsilon^2)}$ to $2^{O(1/\varepsilon)}$. More importantly, the same theorem applies directly to AG codes and yields randomized families of structured linear codes over constant-sized fields that approach the half-Singleton bound.

\section{Notation and preliminaries}
In the entirety of this work we use $\F_q$ to denote the finite field with $q$ elements, and we write $[n] = \{1, \dots, n\}$. 
An $[n,k,d]_q$ linear code $\mathcal{C}$ is defined as a subspace of $\F_q^n$, parameterized by its length $n$, dimension $k$, and minimum Hamming distance $d$. These three parameters of a code are related by the well-known Singleton bound, $d \leq n - k + 1$. Codes meeting this bound are called maximum distance separable (MDS) codes. We may simply refer to an $[n,k]$ code, when the field size and minimum distance are either not relevant or immediate from the context.

\subsection{Preliminaries on insertion/deletion errors} 

We begin by recalling some basic notions concerning insdel errors. The edit distance serves as the fundamental quantity for measuring the distance between two strings. 

\begin{definition}[Edit distance]
	Given two strings $s$ and $s'$, the edit distance $d_{\mathrm{edit}}(s,s')$ is the minimum number of insertions/deletions required to convert $s$ to $s'$ (or vice-versa). \footnote{The edit distance is often denoted by $\mathrm{ED}(\cdot, \cdot)$ or $d_I(\cdot, \cdot)$ in the literature.}  
\end{definition}

The longest common subsequence (LCS) of two strings, which we now define, is closely linked to their edit distance.

\begin{definition}[Longest common subsequence]
	Given two strings $s$ and $s'$, the longest common subsequence of the two strings is a subsequence of both $s$ and $s'$, of maximal length.
	\\$\mathrm{LCS}(s,s')$ denotes the length of this subsequence.
\end{definition}

The following proposition establishes this link between the edit distance and the LCS, allowing us to move interchangeably between the two.

\begin{proposition}
	\label{ED-LCS}
	Suppose $s$ and $s'$ are two strings then $d_{\mathrm{edit}}(s,s') = |s| + |s'| - 2\mathrm{LCS}(s,s')$.	
\end{proposition}
\begin{proof} 
	We can transform $s$ to $s'$ by performing $|s| - \mathrm{LCS}(s,s')$ deletions followed by $|s'| - \mathrm{LCS}(s,s')$ insertions. Therefore, we have an upper bound $d_{\mathrm{edit}}(s,s') \leq |s| + |s'| - 2\mathrm{LCS}(s,s')$.
	Suppose $d_{\mathrm{edit}}(s,s') = D + I$ where some $D$ deletions and $I$ insertions transform $s$ to $s'$. Let $\ell$ be the subsequence of $s$ after $D$ deletions. Then $\ell$ must be a common subsequence with $s'$, thus $\mathrm{LCS(s,s')} \geq |s| - D$. 
	Now, $\ell$ with $I$ insertions gives us $s'$, hence $\mathrm{LCS}(s,s') \geq |s'| - I$. Combining the two we get $2\mathrm{LCS}(s,s') \geq |s| + |s'| - D - I = |s| + |s'| - d_{\mathrm{edit}}(s,s')$. This gives us the lower bound $d_{\mathrm{edit}}(s,s') \geq |s| + |s'| - 2\mathrm{LCS}(s,s')$. 
	The upper and lower bounds give us equality.
\end{proof}

Next, we recall the consequence that linearity imposes on the rate of a code that can correct from a single insertion or deletion. This consequence was proved in \cite{abdel2007linear}, but we give the proof here for completeness.

\begin{lemma}
	An $[n,k]$ linear code $C$ with rate $\frac{k}{n} > \frac{1}{2}$ cannot correct a single insertion/deletion.
\end{lemma}

\begin{proof}
	We want to show that there exist distinct codewords $c,c' \in C$ such that the edit distance $d_{\mathrm{edit}}(c,c') < 3$ or equivalently, by the previous proposition, $\mathrm{LCS}(c,c') \geq n-1$.
	Let $G$ be the generator matrix for an $[n,k]$ linear code with rate $\frac{k}{n} > \frac{1}{2}$.
	Denote by $G^{(i)}$ the matrix $G$ with the $i^{th}$ column removed, $i \in [n]$. Denote by $\langle M \rangle$ the vector space spanned by the rows of a matrix $M$. 
	For any $i,j \in [n]$ we have $$\dim (\langle G^{(i)}\rangle  + \langle G^{(j)}\rangle) \leq \min(n-1,2k) < 2k - 1,$$ since $\frac{k}{n} > \frac{1}{2}$. This implies that there exist nonzero vectors $x$ and $y$ such that $xG^{(i)} - yG^{(j)} = 0$. 
	If we can show that $x \neq y$ for some pair $i,j$ with $i \neq j$, then we will have found two distinct codewords $c = xG$ and $c' = yG$ such that $\mathrm{LCS}(c,c') \geq n-1$.
	Assume that the only nonzero solutions of $xG^{(1)} - yG^{(n)} = 0$ are of the form $x = y$, this implies that $xG$ must be a constant codeword. This tells us that $$\dim (\langle G^{(1)}\rangle  + \langle G^{(n)}\rangle) \geq 2k-1.$$ 
    This is a contradiction and therefore proves the existence of distinct codewords with $\mathrm{LCS}(c,c') \geq n - 1$. 
\end{proof}

A direct consequence of the above lemma gives us the well-established non-asymptotic version of the half-Singleton bound \cite{cheng2020efficient}, which we prove for completeness.

\begin{theorem}[Non-asymptotic half-Singleton bound]
    An $[n,k]$ linear code can correct at most $n-2k + 1$ insdel errors.
\end{theorem}

\begin{proof}
    Suppose the $[n,k]$ linear code $C$ can correct $t$ insdel errors, then for any pair of distinct codewords $c_1,c_2 \in C$ we must have that $d_{\mathrm{edit}}(c_1,c_2) \geq 2t + 1$. Puncture the code $C$ at any $t - 1$ positions to give $C'$. Then $d_{\mathrm{edit}}(c_1',c_2') \geq 2t + 1 - 2(t-1) = 3$ for any pair of distinct codewords $c_1',c_2' \in C'$. This implies that $C'$ can correct one insertion/deletion. 
    The previous lemma implies that $\frac{k}{n-t+1} \leq \frac{1}{2}$, and thus $t \leq n - 2k + 1$.
\end{proof}

For the remainder of the paper it is more convenient to formulate the insdel error correction  guarantees in terms of the longest common subsequence rather than the edit distance. The following standard lemma gives us the equivalent condition to shift our perspective to the LCS.

\begin{lemma}
    \label{LCSrequirement}
    A code $C$ can correct a $\delta$-fraction, that is $\delta n$ insdel errors if and only if $\mathrm{LCS}(c,c') \leq n - \delta n - 1$ for any distinct $c,c' \in C$.
\end{lemma}

\begin{proof}
    A code can correct $\delta n$ insdel errors if and only if $\delta n \leq \frac{d_{\mathrm{edit}}(c,c')-1}{2}$ for any distinct $c,c' \in C$. By Proposition \ref{ED-LCS} we directly get the requirement $\mathrm{LCS}(c,c') \leq n - \delta n - 1$.
\end{proof}

\begin{definition}
    The tuple $I = (I_1, \dots, I_\ell) \in [n]^\ell$ is called an increasing subsequence of length $\ell$ if $I_1 < \cdots < I_\ell$.
\end{definition}

\subsection{Algebraic preliminaries}\label{AGprelim} 
In this section we briefly recall some necessary concepts from algebraic geometry. A complete, in-depth review can be found in \cite{hoholdt1998ag}. Denote by $\bar\F_q$ the algebraic closure of $\F_q$. 
We denote by $\bar\F_q[X_1, \dots, X_n]$ the multivariate polynomial ring and use $\bar\F_q(X_1, \dots, X_n)$ to denote the field of rational functions in $X_1, \dots, X_n$ over $\bar\F_q$. A projective variety $\mathcal{X}$ is the set of zeros in the projective space $\mathbb{P}\bar\F_q^{n-1}$ cut out by a homogeneous prime ideal $\mathcal{I}$ in $\bar\F_q[X_1, \dots, X_n]$. We denote by $\bar\F_q[\mathcal{X}] = \bar\F_q[X_1, \dots, X_n]/\mathcal{I}$ the homogeneous coordinate ring of $\mathcal{X}$. This ring is an integral domain, and the function field $\bar\F_q(\mathcal{X})$ of $\mathcal{X}$ is the field consisting of ratios $f/g$ where $f, g \in \bar\F_q[\mathcal{X}]$ are homogeneous polynomials of the same degree and $g \neq 0$. The transcendence degree of $\bar\F_q(\mathcal{X})$ over $\bar\F_q$ is the dimension of $\mathcal{X}$.

Let $\mathcal{X}$ be a curve defined over $\F_q$, that is, the defining ideal $\mathcal{I}$ can be generated by polynomials with coefficients in $\F_q$. We then write $\F_q(\mathcal{X})$ for the corresponding subfield of $\bar\F_q(\mathcal{X})$ defined over $\F_q$, and work with it throughout. Points on $\mathcal{X}$ with all their coordinates in $\F_q$ are called rational points.

A place $P$ of $\F_q(\mathcal{X})$ is the maximal ideal of a discrete valuation ring $\mathcal{O}_P \subseteq \F_q(\mathcal{X})$. 
A place is called a rational place of $\F_q(\mathcal{X})$ if $\deg(P) := [\mathcal{O}_P/P : \F_q] = 1$.

The rational places of $\F_q(\mathcal{X})$ are in one-to-one correspondence with the rational points of $\mathcal{X}$. We write $\mathcal{P}_q(\mathcal{X})$ for the set of all rational places. For a given place $P$, $v_P$ denotes the corresponding discrete valuation map. For a function $f \in \F_q(\mathcal{X})$ and place $P \in \mathcal{P}_q(\mathcal{X})$, we say $P$ is a zero of $f$ if $v_P(f) > 0$, and $P$ is a pole of $f$ if $v_P(f) < 0$. The discrete valuation $v_P$ measures the multiplicity of a zero or pole of $f$ at $P$.

A divisor of $\mathcal{X}$ is a formal sum $G = \sum_P n_P P$ of places, with $n_P = 0$ for all but finitely many $P$. The degree of $G$ is $\deg(G) := \sum_P n_P \deg(P)$, and we write $G \geq 0$ if all $n_P \geq 0$. 
Since every function $f \in \F_q(\mathcal{X})$ has finitely many poles and zeros, we can introduce the zero and pole divisors of $f$ as $(f)_0 := \sum_{v_P(f)>0} v_P(f) P$ and $(f)_\infty := -\sum_{v_P(f)<0} v_P(f) P$ respectively. The principal divisor of $f$ is $(f) := (f)_0 - (f)_\infty$. Principal divisors have degree zero.

For a divisor $G$, the Riemann--Roch space is
$$
\mathcal{L}(G) := \{f \in \F_q(\mathcal{X})\setminus\{0\} \mid (f) + G \geq 0\} \cup \{0\},
$$
a finite-dimensional $\F_q$-vector space of dimension $\ell(G)$.

We now restrict our attention to algebraic curves $\mathcal{X}$, i.e., algebraic varieties of dimension one.
The genus $g$ of $\mathcal{X}$ is the non-negative integer determined by the Riemann--Roch theorem: for any divisor $G$,
$$
\ell(G) = \deg(G) + 1 - g + \ell(W - G),
$$
where $W$ is a canonical divisor, defined as the divisor of any nonzero rational differential form on $\mathcal{X}$, and satisfying $\deg(W) = 2g - 2$. Two consequences of this theorem are:
\begin{itemize}
    \item $\ell(G) \geq \deg(G) + 1 - g$, with equality whenever $\deg(G) \geq 2g - 1$.
    \item A nonzero $f \in \mathcal{L}(G)$ has at most $\deg(G)$ zeros, counted with multiplicity.
\end{itemize}
We are now ready to define algebraic geometry codes. 
\begin{definition}
    Let $P_1, \dots, P_n$ be distinct rational points on $\mathcal{X}$ and let $G$ be a divisor whose support is disjoint from $\{P_1, \dots, P_n\}$.  The associated algebraic geometry code is
$$
C_{\mathcal{L}(G)}(P_1, \dots, P_n) := \{(f(P_1), \dots, f(P_n)) \mid f \in \mathcal{L}(G)\} \subseteq \F_q^n.
$$
\end{definition}
If $\deg(G) < n$, this is an $[n, k, d]$-code with
$$
k = \ell(G) \geq \deg(G) + 1 - g, \quad d \geq n - \deg(G).
$$
Furthermore, if $2g - 2 < \deg(G) < n$ then
$$
k = \ell(G) = \deg(G) + 1 - g, \quad d \geq n - \deg(G).
$$

We next present a sequence of curves, which give rise to an asymptotically good family of codes.
Let $q = p^2$. The Garcia--Stichtenoth (GS) tower \cite{garcia1996asymptotic} is a sequence of function fields $\mathcal{T}_1 \subseteq \mathcal{T}_2 \subseteq \cdots$ over $\F_q$ defined by $\mathcal{T}_1 = \F_q(X_1)$ and, recursively, $\mathcal{T}_m = \mathcal{T}_{m-1}(X_m)$ where $X_m$ satisfies
$$X_m^p + X_m = \frac{X_{m-1}^p}{X_{m-1}^{p-1} + 1}.$$
Let $\mathcal{Y}_m$ denote the projectivized smooth projective curve with function field $\mathcal{T}_m$. There is unique point $P_\infty$ at infinity, and the number of rational points on the affine part of the curve is 
$$|\mathcal{P}_q(\mathcal{Y}_m) | = (p - 1) p^{m},$$
and the genus is 
$$g_m = \begin{cases} (p^{(m+1)/2} - 1)(p^{(m-1)/2} - 1) & \text{if $m$ is odd} \\ (p^{m/2} - 1)^2 & \text{if $m$ is even} \end{cases}.$$
The ratio $|\mathcal{P}_q(\mathcal{Y}_m)|/g_m \to p - 1 = \sqrt{q} - 1$ as $m \to \infty$, so the tower attains the Drinfeld--Vl\u{a}du\cb{t}  bound \cite{vladut1983bound}, giving an asymptotically good family of codes.

\section{Randomly punctured evaluation codes for insertions/deletions}

In this section, we give a general framework for analyzing the edit distance of random puncturings of evaluation codes. This framework applies to any finite-dimensional space of functions in which any nonzero function has a bounded number of zeros on the evaluation domain. To this end, we first give the general definition of an evaluation code which captures several well-known code families like RS codes and AG codes more generally, RM codes and monomial cartesian codes, to name a few.  

\begin{definition}[Evaluation Code]
    Let $\mathcal{X}$ be a finite evaluation domain that contains a tuple of distinct points $\mathcal{P} = (P_1, \dots, P_n)$, and let $\mathcal{F}$ be a finite-dimensional $\F_q$-vector space of functions
    $f : \mathcal{X} \to \F_q$, which do not have poles on $\mathcal{P}$. 
    The evaluation code associated with $\mathcal{F}$ and $\mathcal{P}$ is
    $$C_{\mathcal{F}}(\mathcal{P}) = \{ (f(P_1), \dots, f(P_n)) \mid f \in \mathcal{F} \} \subseteq \F_q^n.$$
\end{definition}
Throughout this section, we assume that there exists an integer $Z$ such that every nonzero function $f \in \mathcal{F}$ vanishes on at most $Z$ points of $\mathcal{X}$, that is,
$$\left|\lbrace P\in\mathcal{X} \mid f(P)=0\rbrace\right|\leq Z \quad \text{for every nonzero} \quad f \in \mathcal{F}.$$
The parameter $Z$ plays the role of a Schwartz--Zippel bound for the function space $\mathcal{F}$. In particular, if $P$ is chosen uniformly from a subset $S \subseteq \mathcal{X}$, then every nonzero $f \in \mathcal{F}$ satisfies
$$\Pr[f(P)=0]\leq \frac{Z}{|S|}.$$

When specialized to the vector space of functions consisting of polynomials of degree less than $k$ this gives us the RS codes considered in \cite{con2024randomRS} with $Z = k-1$. There the edit distance guarantee is reduced to a necessary and sufficient algebraic condition relying on the rank of particular matrices which they call the $V$-matrix. We generalize this structured $V$-matrix in the below definition for any specified function family which contains the constant function. 

\begin{definition}[V-matrix analogy]
	Given a vector space of functions, $\mathcal{F} = \langle 1, f_1, \dotsm , f_{k-1}\rangle_{\F_q}$ and increasing subsequences $I, J \in [n]^\ell$ of length $\ell$ we define the $V$-matrix as
	$$V_{k,\ell, I, J} = \begin{bmatrix}
	1 & f_1(X_{I_1}) & \cdots & f_{k-1}(X_{I_1}) & f_1(X_{J_1}) & \cdots & f_{k-1}(X_{J_1}) \\
	1 & f_1(X_{I_2}) & \cdots & f_{k-1}(X_{I_2}) & f_1(X_{J_2}) & \cdots & f_{k-1}(X_{J_2}) \\
	\vdots & \vdots & \ddots & \vdots & \vdots & \ddots & \vdots \\
	1 & f_1(X_{I_\ell}) & \cdots & f_{k-1}(X_{I_\ell}) & f_1(X_{J_\ell}) & \cdots & f_{k-1}(X_{J_\ell}) 
	\end{bmatrix} \in \F_q(X_1, \dots, X_n)^{\ell \times (2k - 1)}.$$  
\end{definition}

The $V$-matrix has $\ell$ rows and $2k-1$ columns. Its significance is that a common subsequence between two codewords gives rise to a nonzero vector in its kernel. The following lemma makes this connection precise.

\begin{lemma} \label{notHalfsingleton}
	Let $C_{\mathcal{F}}(\mathcal{P}) \subseteq \F_q^n$ be an $[n,k,d]$ evaluation code with associated vector space $\mathcal{F} = \langle f_0, f_1, \dotsm , f_{k-1}\rangle_{\F_q}$ of functions (with $f_0 = 1$) and evaluation vector $\mathcal{P} = (P_1, \dots, P_n)$. Let $\ell \geq 2k-1$, if $C_{\mathcal{F}}(\mathcal{P})$ cannot correct $n-\ell$ arbitrary insdel errors, then there must exist two increasing subsequences $I,J \in [n]^\ell$ that agree on at most $n-d$ coordinates such that the $V$-matrix $V_{k,\ell,I,J|X = \mathcal{P}}$ evaluated at $\mathcal{P}$ does not have full column rank.
\end{lemma}

\begin{proof}
	Suppose $C_{\mathcal{F}}(\mathcal{P})$ cannot correct $n-\ell$ insdel errors. Then, by Lemma \ref{LCSrequirement} there must exist two codewords $c = (f(P_1), \dots, f(P_n))$ and $c' = (g(P_1), \dots, g(P_n))$ such that $\mathrm{LCS}(c,c') \geq \ell$, where $f = \sum_{i = 0}^{k-1} a_if_i(X)$ and $g = \sum_{i = 0}^{k-1} b_if_i(X)$ are functions in $\mathcal{F}$. The existence of such a longest common subsequence implies the existence of two increasing subsequences $I,J \in [n]^\ell$ such that $$f(P_{I_s}) = g(P_{J_s}) \textit{ for all } s \in \{1, \dots ,\ell\}.$$
	The subsequences can agree on at most $n-d$ coordinates, otherwise we would have $f(P_i) - g(P_i) = 0$ for $n-d+1$ distinct values of $i \in I \cap J$, giving us $d_H(c,c') \leq d-1$ contradicting the fact that $d$ is the minimum distance of the code.
	Consider the vector $v = (a_0 - b_0, a_1, \dots, a_{k-1}, -b_1, \dots, -b_{k-1}) \in \F_q^{2k-1}$. By the existence of a longest common subsequence determined by $I$ and $J$ we have $V_{k,\ell, I,J | X = \mathcal{P}} \cdot v = 0$, confirming that $V_{k,\ell, I,J | X = \mathcal{P}}$ does not have full column rank.	      
\end{proof}

Having established the sufficient general algebraic condition controlling the edit distance of a random evaluation code we can now focus on bounding the probability that a given $V$-matrix is rank deficient. The main idea is to expose the $V$-matrix row by row, checking whether the newly exposed row is in the span of the rows which lie above it. If there is an excess of such linearly dependent rows then we ascertain that the $V$-matrix does not have full column rank. To conveniently represent the rows lying above we use the following notation.

\begin{definition}
    For a matrix $M$ define $M^{|r}$ to be the sub-matrix consisting of the top $r$ rows of the matrix $M$. 
\end{definition}

\begin{lemma}\label{freshvariables}
    The $r^{th}$ row of the matrix $V_{k,\ell,I,J}$ exposes a fresh variable which is not exposed in $V_{k,\ell,I,J}^{|r-1}$.
\end{lemma}
\begin{proof}
    Suppose for the sake of contradiction that the $r^{th}$ row does not expose a fresh variable. That is, there exist $s,t \in \{1, \dots, r-1\}$ such that $X_{I_r} = X_{J_s}$ and $X_{J_r} = X_{I_t}$.
    Since $I$ is an increasing subsequence we have $I_r > I_t$ which implies $I_r = J_s > J_r = I_t$. Since $s < r$ this contradicts the fact that $J$ is an increasing subsequence. Therefore the $r^{th}$ row exposes a fresh variable. 
\end{proof}

We call the $r^{th}$ row of $V_{k,\ell,I,J}$ an \emph{agreement row} if $I_r=J_r$. By Lemma \ref{LCSrequirement}, and the argument of Lemma~10, there are at most $n-d$ such rows. It will be convenient to expose these rows first.

\begin{definition}
    Let $\widetilde{V}_{k,\ell,I,J}$ be the matrix obtained from $V_{k,\ell,I,J}$ by moving all agreement rows to the top, while preserving the relative order of the remaining rows.
\end{definition}

Notice that the conclusion of Lemma \ref{freshvariables} continues to hold for $\widetilde{V}_{k,\ell,I,J}$. Indeed, if $I_r = J_r$, then the variable $X_{I_r}=X_{J_r}$ appears in no other row of $V_{k,\ell,I,J}$, since both $I$ and $J$ are strictly increasing. Consequently, moving agreement rows to the top cannot expose a variable belonging to any of the remaining rows, and every row of $\widetilde{V}_{k,\ell,I,J}$ still contains a fresh variable relative to the rows preceding it.

\begin{lemma}\label{rowinaboverows}
    Let $I,J \in [n]^\ell$ be increasing subsequences that agree on at most $n - d$ coordinates. The probability that the $r^{th}$ row of $\widetilde{V}_{k,\ell,I,J|X = \mathcal{P}}$, with $r > n-d$, is in $W = \mathrm{rowspan}(\widetilde{V}_{k,\ell,I,J|X = \mathcal{P}}^{|r-1})$ with $\dim(W) < 2k-1$, is bounded from above by $\left(\frac{Z}{|\mathcal{X}| - n + 1}\right)$.
\end{lemma}
\begin{proof}
    Let $\rho(X,Y):= [1, f_1(X), \dots, f_{k-1}(X), f_1(Y), \dots, f_{k-1}(Y)]$.
    Since $r > n - d$ the $r^{th}$ row of $\widetilde{V}_{k,\ell,I,J}$ does not have any agreements and is therefore symbolically of the form $$\rho(\alpha, X), \rho(X,\alpha) \textit{ or } \rho(X,Y),$$
    where $X,Y$ are fresh variables exposed by the row that are not already assigned in $\widetilde{V}_{k,\ell,I,J|X = \mathcal{P}}^{|r-1}$ as guaranteed by Lemma \ref{freshvariables} and $\alpha$ is an evaluation point which has already been assigned in $\widetilde{V}_{k,\ell,I,J|X = \mathcal{P}}^{|r-1}$.
    \\
    We restrict our attention to the case $\rho(\alpha, X)$. Let $W = \mathrm{rowspan}(\widetilde{V}_{k,\ell,I,J|X = \mathcal{P}}^{|r-1})$ with $\dim(W) < 2k-1$, then $\rho(\alpha,X) \in W$ if and only if $\rho(\alpha, X)\cdot \omega = 0$ for all $\omega \in W^\bot$. Writing $\omega \in W^{\bot}$ as $\omega = [c_0, \dots, c_{k-1}, d_1, \dots, d_{k-1}] \in W^{\bot}$ we have 
    $$L_\omega(\alpha, X):=  \rho(\alpha, X) \cdot \omega = \sum_{i = 0}^{k-1}c_i f_i(\alpha) + \sum_{j = 1}^{k-1}d_j f_j(X).$$
    Now, $\rho(\alpha,X) \in W \implies L_{\omega}(\alpha, X) = 0$, therefore $\Pr[\rho(\alpha,X) \in W] \leq \Pr[L_\omega(\alpha, X) = 0]$ for all $\omega \in W^\bot$. To ensure that this bound is useful we need to show the existence of $\omega_0 \in W^\bot$ such that $L_{\omega_0}(\alpha, X)$ is not identically zero. 
    \\
    Suppose for the sake of contradiction that $L_{\omega}(\alpha, X) \equiv 0$ for any $\omega = [c_0, \dots, c_{k-1}, d_1, \dots, d_{k-1}] \in W^\bot$. This would mean that $d_j = 0$ for all $j \in \{1, \dots, k-1\}$ since $f_1(X), \dots, f_{k-1}(X)$ are linearly independent. Furthermore $L'_{\omega}(\alpha) := \sum_{i = 0}^{k-1}c_i f_i(\alpha) = 0$ and $L'_\omega(P_{I_s}) = 0$ for all $s \in \{1, \dots, r-1\}$. Notice that $[L'_\omega(P_1), \dots, L'_\omega(P_n)]$ is a codeword which under our assumption has $r > n - d$ zeros, the only such codeword is the zero codeword and we get $c_i = 0$ for all $i \in \{0, \dots, k-1\}$. Therefore $\omega = 0$ and
    $$\dim(W^\bot) = 0.$$  
    But we know that
    $$\dim(W^\bot) = 2k-1 - \dim(W).$$
    This forces $\dim(W) = 2k-1$ which is a contradiction. Therefore there exists an $\omega_0 \in W^\bot$ such that $L_{\omega_0} \not\equiv 0$. 
    \\
    Now, since $L_\omega(\alpha, X) \in \mathcal{F}$, it must by assumption vanish on at most $Z$ points of $\mathcal{X}$. The freshly exposed point $X$ is uniform amongst all the unused evaluation points for which there are at least $|\mathcal{X}| - n + 1$ possibilities, therefore we have 
    $$\Pr[\rho(\alpha,X) \in W] \leq \Pr[L_\omega(\alpha, X) = 0] \leq \left(\frac{Z}{|\mathcal{X}| - n + 1}\right).$$
    The argument for $\rho(X, \alpha)$ follows in a symmetric fashion to the argument for $\rho(\alpha, X)$ and the argument for $\rho(X,Y)$ involves simply assigning $Y$ to reduce to the $\rho(X, \alpha)$ case. 
\end{proof}

\begin{lemma} \label{Probnotfullrank}
    Let $I,J \in [n]^\ell$ be increasing subsequences of length $\ell$ that agree on at most $n-d$ coordinates. Then the probability that $V_{k, \ell, I, J|X=\mathcal{P}}$ does not have full column rank is bounded from above by 
    $$\binom{\ell - (n - d + 1)}{\tau} \left(\frac{Z}{|\mathcal{X}| - n + 1}\right)^\tau , $$
    where $\tau  := \ell - k - (n - d) + 1$.
\end{lemma}
\begin{proof}
    The (transpose of) the leftmost $n - d + 1 \times k$ minor of the matrix $\widetilde{V}_{k, \ell, I, J|X = \mathcal{P}}^{|n-d + 1}$ can be seen as a puncturing of the generator matrix of the code onto $n - d + 1$ columns. We know that this punctured code must have dimension $k$ and therefore $\mathrm{rank}(\widetilde{V}_{k, \ell, I, J|X = \mathcal{P}}^{|n-d + 1}) \geq k$. For our bound it is enough to consider the worst case when $\mathrm{rank}(\widetilde{V}_{k, \ell, I, J|X = \mathcal{P}}^{|n-d + 1}) = k$. 
    In that case the amount of slack remaining is $$\tau - 1 = (\ell - (2k-1)) - ((n-d+1)- k) = \ell - k - (n-d).$$
    Let $(\rho_{r_1}, \dots, \rho_{r_\tau})$ with $n - d + 1 < r_1 < \cdots < r_\tau \leq \ell$ be a collection of rows of $\widetilde{V}_{k,\ell, I, J}$. Then by Lemma \ref{rowinaboverows} we have $$\Pr[\rho_{r_1} \in W_{r_1}\land \cdots \land \rho_{r_\tau} \in W_{r_\tau}] \leq \left(\frac{Z}{|\mathcal{X}| - n + 1}\right)^\tau,$$
    where $W_{r_i} := \mathrm{rowspan}(\widetilde{V}_{k,\ell,I,J|X = \mathcal{P}}^{|r_i - 1})$. The matrix $\widetilde{V}_{k, \ell, I, J|X=\mathcal{P}}$ does not have full column rank only if there exist $\tau$ rows $(\rho_{r_1}, \dots, \rho_{r_\tau})$ such that $\rho_{r_1} \in W_{r_1}\land \cdots \land \rho_{r_\tau} \in W_{r_\tau}$. Taking the union bound over all possible $(\rho_{r_1}, \dots, \rho_{r_\tau})$ with $n - d + 1 < r_1 < \cdots < r_\tau \leq \ell$ we get 
    $$\binom{\ell - (n - d + 1)}{\tau} \left(\frac{Z}{|\mathcal{X}| - n + 1}\right)^\tau.$$
\end{proof}

\begin{theorem} \label{ultimate}
    Let $C_{\mathcal{F}}(\mathcal{P}) \subseteq \F_q^n$ be an $[n,k,d]$ evaluation code with associated vector space $\mathcal{F} = \langle f_0, f_1, \dotsm , f_{k-1}\rangle_{\F_q}$ of functions (with $f_0 = 1$) and evaluation vector $\mathcal{P} = (P_1, \dots, P_n)$. Let $\ell \geq 2k-1$, then
    $$\Pr[C_{\mathcal{F}}(\mathcal{P}) \textit{ cannot correct $n-\ell$ insdel errors}] \leq \binom{n}{l}^2 \binom{\ell - (n - d + 1)}{\tau} \left(\frac{Z}{|\mathcal{X}| - n + 1}\right)^\tau,$$
    where $\tau  := \ell - k - (n - d) + 1$.
\end{theorem}
\begin{proof}
    By Lemma \ref{notHalfsingleton} if $\mathcal{C_F(P)}$ cannot correct $n-\ell$ insdel errors then there exist two increasing subsequences $I,J \in [n]^\ell$ that agree on at most $n-d$ coordinates such that $V_{k,\ell, I,J | X = \mathcal{P}}$ does not have full column rank. Lemma \ref{Probnotfullrank} bounds the probability that a specific $V$-matrix does not have full column rank. The number of admissible pairs $(I,J)$ is at most $\binom{n}{\ell}^2$. Therefore, taking a union bound over all such pairs gives us the desired bound,
    $$\Pr[C_{\mathcal{F}}(\mathcal{P}) \textit{ cannot correct $n-\ell$ insdel errors}] \leq \binom{n}{l}^2 \binom{\ell - (n - d + 1)}{\tau} \left(\frac{Z}{|\mathcal{X}| - n + 1}\right)^\tau.$$
    A crude bound on the binomial coefficients gives us a less accurate but much simpler bound,
    $$\Pr[C_{\mathcal{F}}(\mathcal{P}) \textit{ cannot correct $n-\ell$ insdel errors}] \leq 2^{3n}\left(\frac{Z}{|\mathcal{X}| - n + 1}\right)^\tau.$$
\end{proof}

\section{Specializing to specific families of evaluation codes}

In this section we instantiate Theorem \ref{ultimate} to specific families of evaluation codes. First, we demonstrate the generality of our theorem by considering the vector space of $m$-variate polynomials of degree at most $r$ with coefficients from the finite field $\F_q$ which gives rise to the family of $q$-ary RM codes. We obtain a constant-fraction guarantee for first-order binary RM codes, complementing the work of \cite{qiu2026binary} on insdel-robust subcodes of binary first-order RM codes. Second, we look at RS codes and recover the linear field size requirement of \cite{con2024randomRS} with an asymptotic improvement in terms of the additive gap $\varepsilon$. Third, we push the limits on the field size by instantiating the general theorem with an asymptotically good family of AG codes arising from the GS tower of function fields. This ultimately gives us a structured family of codes capable of almost attaining the half-Singleton bound over constant-sized fields.

\subsection{Reed--Muller codes}
\begin{definition}[$q$-ary Reed--Muller code]
    For a prime power $q$ and integers $m \geq 1$ and $r < q$, let $\mathcal{F} = \lbrace f \in \F_q[X_1,\dots,X_m] \mid \deg(f)\leq r\rbrace$. Let $\mathcal{P} = (P_1, \dots, P_{q^m})$ be the points of $\F_q^m$ ordered lexicographically. The $q$-ary Reed--Muller code of order $r$ in $m$ variables is defined as
    $$C_{\mathcal{F}}(\mathcal{P}) = \mathrm{RM}_q(m,r) = \lbrace \left(f(P_1),\dots,f(P_{q^m})\right) \mid f \in \mathcal{F} \rbrace.$$ 
    It is a linear $[q^m,k,d]_q$ code with dimension $k=\binom{m+r}{m}$ and minimum distance $d = (q-r)q^{m-1}$.
\end{definition}

\begin{corollary} \label{RMq}
    Let $\varepsilon \in (0,1)$, $m > 0$. Suppose that 
    $$k+ \lfloor\varepsilon n\rfloor > rq^{m-1} \quad \text{and} \quad \frac{rq^{m-1}}{q^m-n+1} \leq 2^{-\left(\frac{4n}{\,k+\lfloor\varepsilon n\rfloor-rq^{m-1}}\right)}.$$
    Then with probability at least $1-2^{-n}$ a uniformly random ordered length-$n$ puncturing of $\mathrm{RM}_q(m,r)$ corrects at least $(1 - \varepsilon)n-2k+1$ insdel errors.
\end{corollary}

\begin{proof}
    We apply the general Theorem \ref{ultimate} with $\mathcal{X}=\mathbb{A}^m(\F_q)=\F_q^m$ and the function space $\mathcal{F}$ of $m$-variate polynomials of total degree at most $r$. By Schwartz-Zippel, every nonzero $f \in \mathcal{F}$ vanishes on at most $Z = rq^{m-1}$ points of $\mathcal{X}$. This gives us $n-d \leq Z = rq^{m-1}$. Taking $\ell=2k - 1 + \lfloor\varepsilon n\rfloor$, we have 
    $$\tau = \ell-k-(n-d)+1 \geq k + \lfloor\varepsilon n\rfloor -rq^{m-1} > 0.$$
    By Theorem \ref{ultimate} we have, 
    \begin{align*}
        \Pr[C_{\mathcal{F}}(\mathcal{P}) \textit{ cannot correct $n-2k + 1 - \lfloor \varepsilon n \rfloor$ insdel errors}] &\leq 2^{3n}\left(\frac{Z}{q^m-n+1}\right)^\tau \\
        &\leq 2^{3n}\left(\frac{rq^{m-1}}{q^m-n+1}\right)^\tau \leq 2^{-n},
    \end{align*}
    where the final inequality follows from the hypothesis. Therefore, with probability at least $1-2^{-n}$, the random ordered puncturing corrects $n-\ell = n-2k+1-\lfloor\varepsilon n\rfloor \geq (1 - \varepsilon)n - 2k + 1$ insdel errors.
\end{proof}

\begin{remark}
    The hypotheses in Corollary \ref{RMq} are satisfied in a particular regime. Concretely, for a fixed $\delta \in (0,1)$ choose a sufficiently small $\gamma \in (0,1)$ such that 
    $$\frac{\gamma^m}{m!}+ \varepsilon\delta > \gamma \quad \text{and} \quad \frac{\gamma}{1-\delta} < 2^{-\left( \frac{4\delta}{\frac{\gamma^m}{m!}+\varepsilon\delta-\gamma}\right)}.$$
    Notice that such a choice of $\gamma$ always exists. Let $n=\lfloor\delta q^m\rfloor$ so that $q = O(n^{1/m})$ and take $r=\lfloor\gamma q\rfloor$. It is easy to verify that under these conditions the required hypotheses are met. These conditions constrain the randomly punctured RM codes to a low-rate regime. In conclusion, for every
    $$R\in\left(0,\frac{\gamma^m}{m!\delta}\right),$$ 
    Corollary \ref{RMq} yields families of codes over fields of size $q=O(n^{1/m})$ that can correct at least $(1 - \varepsilon)n - 2k + 1$ insdel errors.
\end{remark}

We can also derive another interesting consequence of  Theorem \ref{ultimate} by considering first-order binary RM codes. When $m = o(n)$ we show that a uniformly random ordered length-$n$ puncturing of $\mathrm{RM}_2(m,1)$ can correct a linear number $\delta n$ of insdel errors. To arrive at this result we first present the following proposition which shows that a random puncturing has minimum distance close to $n/2$ with high probability. 

\begin{proposition} \label{RM2mindist}
    Let $m,n$ be positive integers with $m = o(n)$ and fix $\mu \in \left(0,\frac{1}{2}\right)$. Suppose $\mathcal{P} = (P_1, \dots, P_n)$ is chosen uniformly at random from the set of all $n$-tuples of distinct points in $\F_2^m$. Then, with probability at least $1 - 2^{-\Omega(\mu^2n)}$, the random ordered length-$n$ puncturing $\mathcal{C_F(P)}$ of $\mathrm{RM}_2(m,1)$ has minimum Hamming distance $d > \left( \frac{1}{2} - \mu\right)n$.
\end{proposition}
\begin{proof}
    The function space underlying the $\mathrm{RM}_2(m,1)$ is $\mathcal{F} = \langle 1, X_1, \dots, X_m\rangle_{\F_2}$. Every non-constant function $f \in \mathcal{F}$ evaluates to $1$ on exactly $2^{m-1}$ points of $\F_2^m$. We can compute for any non-constant function $f \in \mathcal{F}$ the following probability 
    $$\Pr\left[\mathrm{wt}(f(P_1), \dots, f(P_n)) \leq  \left( \frac{1}{2} - \mu\right)n\right] = \sum_{i=0}^{\left( \frac{1}{2} - \mu\right)n}\frac{\binom{2^{m-1}}{i}\binom{2^{m-1}}{n-i}}{\binom{2^{m}}{n}} \leq e^{-2\mu^2n}.$$
    Taking a union bound over all the $2^{m+1}-2$ non-constant functions in $\mathcal{F}$ we get 
    $$\Pr\left[d\leq  \left( \frac{1}{2} - \mu\right)n\right] \leq (2^{m+1}-2)e^{-2\mu^2n}.$$
    Since $m = o(n)$ we get $(2^{m+1}-2)e^{-2\mu^2n} = 2^{-\Omega(\mu^2n)}$, therefore with probability at least $1 - 2^{-\Omega(\mu^2n)}$ the random ordered puncturing $\mathcal{C_F(P)}$ of $\mathrm{RM}_2(m,1)$ has minimum Hamming distance $$d > \left( \frac{1}{2} - \mu\right)n.$$
\end{proof}

\begin{corollary}
Let $m,n$ be positive integers satisfying $m=o(n)$ and $n = o(2^m)$. Let $\delta>0$ satisfy
$$2H_2(\delta)<\frac{1}{2}-\delta, \quad \text{equivalently} \quad \delta < \delta',$$
where $H_2$ denotes the binary entropy function and $\delta' \approx 0.0376$. Suppose
$P=(P_1,\dots,P_n)$ is chosen uniformly at random from the set of all $n$-tuples of distinct points in $\F_2^m$. Then, with probability at least $1-2^{-\Omega(n)}$, the random ordered length-$n$ puncturing $C_{\mathcal{F}}(P)$ of $\mathrm{RM}_2(m,1)$ corrects at least $\lfloor\delta n\rfloor$ insdel errors.
\end{corollary}
\begin{proof}
    We apply the general Theorem \ref{ultimate} with $\mathcal{X} = \F_2^m$ and the function space $\mathcal{F} = \langle 1, X_1, \dots, X_m\rangle_{\F_2}$ of affine linear functions in $m$ variables. Every nonzero $f \in \mathcal{F}$ vanishes on at most $Z = 2^{m-1}$ points of $\mathcal{X}$. By Proposition \ref{RM2mindist} we have, with probability at least $1 - 2^{-\Omega(\mu^2n)}$, that the punctured code $C_{\mathcal{F}}(P)$ has minimum distance $d > \left( \frac{1}{2} - \mu\right)n$. Call this event $\mathfrak{D}$, on this event we have $n-d < \left( \frac{1}{2} + \mu\right)n$. Taking $\ell = n - \lfloor \delta n\rfloor$, and since $k = m + 1 = o(n)$, we have
    $$\tau = \ell - k - (n-d) + 1 > \left( \frac{1}{2} - \mu - \delta - o(1)\right)n.$$
    Then, by Theorem \ref{ultimate} we have 
    $$\Pr[C_{\mathcal{F}}(\mathcal{P}) \textit{ cannot correct $\lfloor \delta n\rfloor$ insdel errors} ~\wedge~\mathfrak{D}] \leq 
    \binom{n}{\lfloor \delta n\rfloor}^2 \binom{d - \lfloor \delta n\rfloor - 1}{k-2} \left(\frac{2^{m-1}}{2^m - n + 1}\right)^\tau.$$
    The first factor has the standard bound $\binom{n}{\lfloor \delta n\rfloor}^2  \leq 2^{2H_2(\delta)n}$, while the second factor $\binom{d - \lfloor \delta n\rfloor - 1}{k-2} \leq \binom{n}{k} = 2^{o(n)}$ since $k = o(n)$. The final factor 
    $$\left(\frac{2^{m-1}}{2^m - n + 1}\right)^\tau \leq 2^{-\left( \frac{1}{2} - \mu - \delta - o(1)\right)n}$$
    by the fact that $\tau > \left( \frac{1}{2} - \mu - \delta - o(1)\right)n$ and our assumption $n = o(2^m)$.
    Therefore we have 
    $$\Pr[C_{\mathcal{F}}(\mathcal{P}) \textit{ cannot correct $\lfloor \delta n\rfloor$ insdel errors} ~\wedge~\mathfrak{D}] \leq 2^{-\left( \frac{1}{2} - \mu - \delta - 2H_2(\delta) - o(1)\right)n},$$
    which by our assumption on $\delta$ gives us 
    $$\Pr[C_{\mathcal{F}}(\mathcal{P}) \textit{ cannot correct $\lfloor \delta n\rfloor$ insdel errors} ~\wedge~\mathfrak{D}] = 2^{-\Omega(n)}.$$
    Proposition \ref{RM2mindist} tells us that $\Pr[\mathfrak{D}] \geq 1 - 2^{-\Omega(n)}$ and therefore we have 
    $$\Pr[C_{\mathcal{F}}(\mathcal{P}) \textit{ cannot correct $\lfloor \delta n\rfloor$ insdel errors}] = 2^{-\Omega(n)}.$$
    Since $\delta' \approx 0.0376$ is the unique positive solution for $2H_2(\delta) = \frac{1}{2}-\delta$, we have, with probability at least $1-2^{-\Omega(n)}$, the random ordered length-$n$ puncturing $C_{\mathcal{F}}(\mathcal{P})$ of $\mathrm{RM}_2(m,1)$ corrects at least $\lfloor\delta n\rfloor$ insdel errors for all $\delta \in (0,\delta')$. 
\end{proof}

\subsection{Reed--Solomon codes}
\begin{definition}[Reed--Solomon code]
    Let $q$ be a prime power and let $k$ and $n$ be integers satisfying $0 < k \leq n \leq q$. Let $\mathcal{F} = \lbrace f \in \F_q[X] \mid \deg(f) < k\rbrace$. Let $P = (P_1, \dots, P_{n})$ be a tuple of $n$ distinct points of $\F_q$. The $[n,k]_q$ Reed--Solomon code evaluated on $P$ is defined as
    $$C_{\mathcal{F}}(\mathcal{P}) = \mathrm{RS}_{n,k}(\mathcal{P}) = \lbrace \left(f(P_1),\dots,f(P_{n})\right) \mid f \in \mathcal{F} \rbrace.$$ 
    The code has minimum distance $d = n-k+1$.
\end{definition}

\begin{corollary} \label{RS}
    Let $\varepsilon \in (0,1)$. Let $q$ be a prime power, $q \geq 2^{\frac{4}{\varepsilon}}(k-1) + n$. Suppose $\mathcal{P} = (P_1, \dots, P_n)$ is chosen uniformly at random from the set of all $n$-tuples of distinct elements in $\F_q$. Then with probability at least $1- 2^{-n}$ the Reed--Solomon code $\mathrm{RS}_{n,k}({\mathcal{P}})$ can correct at least $(1 - \varepsilon)n - 2k + 1$ insdel errors.
\end{corollary}
\begin{proof}
    We apply the general Theorem \ref{ultimate} with $\mathcal{X} = \F_q$ and the function space $\mathcal{F} = \langle 1, X , \dots, X^{k-1}\rangle$ of polynomials of degree at most $k-1$. By Schwartz-Zippel, every nonzero $f \in \mathcal{F}$ vanishes on at most $Z = n-d = k-1$ points of $\mathcal{X}$.
    Taking $\ell = 2k - 1 + \lfloor \varepsilon n \rfloor$, we have 
    $$\tau = \ell - k - (n-d) + 1 = 1 + \lfloor\varepsilon n \rfloor > 0.$$
    Then, by Theorem \ref{ultimate} we have 
    $$\Pr[C_{\mathcal{F}}(\mathcal{P}) \textit{ cannot correct $n-2k + 1 - \lfloor \varepsilon n \rfloor$ insdel errors}] \leq 2^{3n}\left(\frac{k-1}{q - n + 1}\right)^{\lceil \varepsilon n \rceil}.$$
    Plugging in $q \geq 2^{\frac{4}{\varepsilon}}(k-1) + n$ we get 
    $$\Pr[C_{\mathcal{F}}(\mathcal{P}) \textit{ cannot correct $n-2k + 1 - \lfloor \varepsilon n \rfloor$ insdel errors}] \leq 2^{3n}\left(\frac{1}{2^{\frac{4}{\varepsilon}}}\right)^{\lceil \varepsilon n \rceil} \leq 2^{-n}.$$
    Therefore, with probability at least $1-2^{-n}$, the random choice of evaluation points $\mathcal{P} \in \mathcal{X}$ yields a code capable of correcting at least $n-\ell = n-2k+1-\lfloor\varepsilon n\rfloor \geq (1 - \varepsilon)n - 2k + 1$ insdel errors.
\end{proof}

\begin{remark}
    We can also arrive at Corollary \ref{RS} by realizing that it is a special case of Corollary \ref{RMq} with $m = 1$ and $r = k-1$. In this case the constraints on the rate disappear. The first hypothesis in Corollary \ref{RMq} is automatic, and the second hypothesis is satisfied by our choice of $q \geq 2^{\frac{4}{\varepsilon}}(k-1) + n$. Therefore every fixed rate $R \in \left(0, \frac{1}{2}\right)$ is admissible. 
\end{remark}
\subsection{Algebraic geometry codes}
In this section we arrive at our main result, showing that there exists a randomized family of linear codes which achieves the half-Singleton bound over constant-sized fields. The particular structured codes that we consider are from a family of AG codes defined over an asymptotically optimal tower of function fields. The preliminaries on this family of AG codes can be found in Section \ref{AGprelim}.

\begin{corollary}
    Let $\varepsilon \in (0,1)$ and $R \in (0, 1/2)$. Let $p$ be a prime power such that 
    $$p \geq 1 + \max\left\lbrace \frac{4}{R}, \frac{8}{\varepsilon}\right\rbrace \left(1 +2^{8/\varepsilon}\right),$$
    for every sufficiently large $n$ for which there exists an integer $m$ such that $$ \min\left\lbrace \frac{Rn}{4}, \frac{\varepsilon n}{8}  \right\rbrace \leq p^m \leq \min\left\lbrace \frac{Rn}{2}, \frac{\varepsilon n}{4}\right\rbrace,$$
    and set $q = p^2$. 
    Suppose $\mathcal{P} = (P_1, \dots, P_n)$ is chosen uniformly at random from the set of all $n$-tuples of distinct elements in $\mathcal{P}_q(\mathcal{Y}_m) \setminus P_\infty$, where $\mathcal{Y}_m$ is the $m^{th}$ curve in the GS tower with genus $g_m$. Let $G = (k - 1 + g_m)P_{\infty}$ be a divisor. Then, with probability at least $1-2^{-n}$ the $[n, k = Rn]$ AG code $\mathcal{C}_{\mathcal{L}(G)}(\mathcal{P})$ can correct at least $(1-\varepsilon)n - 2k + 1$ insdel errors.
\end{corollary}
\begin{proof}
    We apply the general Theorem \ref{ultimate} with $\mathcal{X} = \mathcal{P}_q(\mathcal{Y}_m) \setminus \{P_\infty\}$ and the function space as the Riemann--Roch space $\mathcal{F} =\mathcal{L}(G)$.
    For the genus $g_m$ of $\mathcal{Y}_m$ we have the upper bound $g_m \leq 2p^m$ which along with $p^m \leq \frac{Rn}{2}$ implies $$g_m \leq Rn = k. $$
    This forces $\deg(G) = k + g_m - 1 \geq 2g_m - 1$ which consequently tells us by the Riemann--Roch theorem that the function space $\mathcal{L}(G)$ has dimension $\ell(G) = k$. 
    For any function $f \in \mathcal{F}$ we have that $\deg(G)$ bounds the number of zeros, and so we set our zero bound $Z = \deg(G)$. 
    Take $\ell = 2k - 1 + \lfloor \varepsilon n \rfloor$. The assumed inequality $p^m \leq \frac{\varepsilon n}{4}$ tells us that $g_m \leq \frac{\varepsilon n}{2}$ which forces 
    \begin{align*}
        \tau &:= \ell - k - (n-d) + 1 \\
        & \geq 2k - 1 + \lfloor \varepsilon n \rfloor - k - (k-1+g_m) + 1 \\
        & \geq  \lfloor \varepsilon n \rfloor + 1 - g_m \\
        & \geq \left\lceil\frac{\varepsilon n}{2} \right\rceil .
    \end{align*}
    Then, by Theorem \ref{ultimate} we have 
    $$\Pr[C_{\mathcal{F}}(\mathcal{P}) \textit{ cannot correct $n-2k + 1 - \lfloor \varepsilon n \rfloor$ insdel errors}] \leq 2^{3n}\left(\frac{\deg(G)}{|\mathcal{P}_q(\mathcal{Y}_m) \setminus \{P_\infty\}| - n + 1}\right)^{\frac{\varepsilon n}{2}}.$$
    Note that
    \begin{align*}
        \frac{\deg(G)}{|\mathcal{P}_q(\mathcal{Y}_m) \setminus \{P_\infty\}| - n + 1} &= \frac{Rn + g_m - 1}{p^m(p-1) - n + 1} \\
        & \leq \frac{Rn + \frac{\varepsilon n}{2}  - 1}{p^m(p-1) - n + 1} \qquad \text{(by $g_m \leq \frac{\varepsilon n}{2}$)}\\
        & \leq \frac{n}{p^m(p-1) - n + 1} \qquad \text{(by $\varepsilon \in (0,1)$ and $R \in (0, 1/2)$)}\\
        & \leq \frac{p^m \max\left\lbrace \frac{4}{R}, \frac{8}{\varepsilon}\right\rbrace}{p^m(p-1)- p^m\max\left\lbrace \frac{4}{R}, \frac{8}{\varepsilon}\right\rbrace + 1} \qquad \text{(by $p^m \geq \min\left\lbrace \frac{Rn}{4}, \frac{\varepsilon n}{8}\right\rbrace$)} \\
        & \leq \frac{1}{\min\left\lbrace \frac{R}{4}, \frac{\varepsilon}{8}\right\rbrace(p-1)-1}. 
    \end{align*}
    This simplifies the failure probability to 
    $$\Pr[C_{\mathcal{F}}(\mathcal{P}) \textit{ cannot correct $n-2k + 1 - \lfloor \varepsilon n \rfloor$ insdel errors}] \leq 2^{3n}\left(\frac{1}{\min\left\lbrace \frac{R}{4}, \frac{\varepsilon}{8}\right\rbrace(p-1)-1}\right)^{\frac{\varepsilon n}{2}}.$$
    With the consideration $p \geq 1 + \max\left\lbrace \frac{4}{R}, \frac{8}{\varepsilon}\right\rbrace \left(1 +2^{8/\varepsilon}\right)$ this gives us the desired
    $$\Pr[C_{\mathcal{F}}(\mathcal{P}) \textit{ cannot correct $n-2k + 1 - \lfloor \varepsilon n \rfloor$ insdel errors}] \leq 2^{-n}.$$
    Therefore, with probability at least $1-2^{-n}$, the random choice of evaluation points $\mathcal{P} \in \mathcal{X}$ yields a code capable of correcting at least $n-\ell = n-2k+1-\lfloor\varepsilon n\rfloor \geq (1 - \varepsilon)n - 2k + 1$ insdel errors over a field of size $q  = p^2= 2^{O(1/\varepsilon + \log(1/R))}$.
\end{proof}

\section{Conclusions and future work}
We gave a simple general framework for analyzing random evaluation codes against insdel errors. For Reed--Muller codes, the framework gives new insdel error correction guarantees for random puncturings, including a constant-fraction guarantee for first-order binary Reed--Muller codes. This complements the recent work of Qiu et al. \cite{qiu2026binary} on insdel-robust subcodes of binary first-order Reed--Muller codes, whereas our result gives a constant-fraction guarantee for random puncturings of the full first-order Reed--Muller code. For Reed--Solomon codes, the framework gives a substantially shorter proof of the result of Con et al.\cite{con2024randomRS} and improves the dependence on the additive gap $\varepsilon$ from $2^{O(1/\varepsilon^2)}$ to $2^{O(1/\varepsilon)}$. Most importantly, the same theorem applies directly to AG codes and yields randomized families of structured linear codes over constant-sized fields that approach the half-Singleton bound.

In future work, it would be interesting to look into similar problems for a fixed/restricted $q$, in the same fashion as has been done for, e.g., locally recoverable codes (LRCs) \cite{cadambe2015localbound,grezet2019alphabet,gruica2023duallocal}, and study families of codes over a fixed field size getting (close) to the half-Plotkin  bound. 

The algebraic condition required for a linear code to be uniquely decodable from insertions and deletions can easily be extended to an algebraic condition required for a linear code to be list-decodable from insertions and deletions. Relaxing the unique decoding requirement allows for linear codes capable of beating the half-Singleton bound, as demonstrated in \cite{li2026rateone}. It would be interesting to extend our framework to get randomized constructions of structured list-decodable linear codes with rate approaching $1$ over constant-sized fields. Another interesting direction would be to extend our results to generalized evaluation codes approaching the strict half-Singleton bound in the same vein as the generalization of \cite{liu2024optimal} for Reed--Solomon codes. 

Finally, it remains to find explicit constructions of linear codes and accompanying efficient decoding algorithms, which realize our existence results.  

\bibliographystyle{ieeetr}
\bibliography{main}

\end{document}